\documentclass{article}

\usepackage{ijcai13, amsthm, amssymb, amsmath}
\usepackage[ruled,vlined]{algorithm2e}
\usepackage{graphicx}
\usepackage{booktabs}
\usepackage{color}

\usepackage{times}

\usepackage{tikz}
\usetikzlibrary{arrows,decorations,decorations.shapes,backgrounds,shapes}

\newtheorem{lemma}{Lemma}
\newtheorem{theorem}{Theorem}
\newenvironment{remark}[1][Remark]{\begin{trivlist}
\item[\hskip \labelsep {\bfseries #1}]}{\end{trivlist}}

\newcount\Comments
\definecolor{darkgreen}{rgb}{0,0.7,0}
\newcommand{\kibitz}[2]{\ifnum\Comments=1\textcolor{#1}{#2}\fi}

\newcommand{\cut}[1]{}

\newif\ifbetterbound
\betterboundtrue
\newcommand{\newbound}[2]{%
\ifbetterbound%
#2%
\else%
#1%
\fi}

\newif\iffullversion
\title{Audit Games
\thanks{This work was partially supported by the U.S. Army Research Office contract “Perpetually
Available and Secure Information Systems” (DAAD19-02-1-0389) to Carnegie Mellon
CyLab, the NSF Science and Technology Center TRUST, the NSF CyberTrust grant “Privacy,
Compliance and Information Risk in Complex Organizational Processes,” the AFOSR
MURI “Collaborative Policies and Assured Information Sharing,”, HHS Grant no. HHS
90TR0003/01 and NSF CCF-1215883. Jeremiah Blocki was also partially supported by a NSF Graduate Fellowship.
Arunesh Sinha was also partially supported by the CMU CIT Bertucci Fellowship. The views
and conclusions contained in this document are those of the authors and should not be interpreted
as representing the official policies, either expressed or implied, of any sponsoring
institution, the U.S. government or any other entity. }
}
\author{Jeremiah Blocki, Nicolas Christin, Anupam Datta, Ariel D. Procaccia, Arunesh Sinha \\
Carnegie Mellon University, Pittsburgh, USA \\
\{jblocki@cs, nicolasc@, danupam@, arielpro@cs, aruneshs@\}.cmu.edu}

\begin{document}

\maketitle

\begin{abstract}
Effective enforcement of laws and policies requires expending resources to
prevent and detect offenders, as well as appropriate punishment schemes to
deter violators. In particular, enforcement of privacy laws and policies in
modern organizations that hold large volumes of personal information (e.g.,
hospitals, banks, and Web services providers) relies
heavily on internal audit mechanisms. We study economic considerations in the
design of these mechanisms, focusing in particular on effective resource
allocation and appropriate punishment schemes. We present an audit game model
that is a natural generalization of a standard security game model for resource
allocation with an additional punishment parameter. Computing the Stackelberg
equilibrium for this game is challenging because it involves solving an
optimization problem with non-convex quadratic constraints.  We present an
additive FPTAS that efficiently computes a solution that is arbitrarily close
to the optimal solution.
\end{abstract}

\section{Introduction}

In a seminal paper, Gary Becker~\shortcite{Becker68} presented a compelling
economic treatment of crime and punishment. He demonstrated that        
effective law enforcement involves optimal resource allocation to       
prevent and detect violations, coupled with appropriate punishments     
for offenders. He described how to optimize resource allocation by      
balancing the societal cost of crime and the cost incurred by 
prevention, detection and punishment schemes. While 
Becker focused on crime and punishment in society, similar economic     
considerations guide enforcement of a wide range of policies. In this   
paper, we study effective enforcement mechanisms for this broader       
set of policies. Our study differs from Becker's in two significant     
ways---our model accounts for \emph{strategic interaction} between      
the enforcer (or defender) and the adversary; and we design efficient   
algorithms for \emph{computing} the optimal resource allocation for     
prevention or detection measures as well as punishments. At the same    
time, our model is significantly less nuanced than Becker's, thus       
enabling the algorithmic development and raising interesting questions  
for further work.                                                       

A motivating application for our work is auditing, which typically involves
detection and punishment of policy violators. In particular, enforcement of 
privacy laws and policies in modern organizations that hold large volumes of 
personal information (e.g., hospitals, banks, and Web services providers like 
Google and Facebook) relies heavily on internal audit mechanisms. Audits are 
also common in the financial sector (e.g., to identify fraudulent transactions),
in internal revenue services (e.g., to detect tax evasion), and in traditional law 
enforcement (e.g., to catch speed limit violators). 

The audit process is an interaction between two agents: a defender (auditor)
and an adversary (auditee). As an example, consider a hospital (defender)
auditing its employee (adversary) to detect privacy violations committed by the
employee when accessing personal health records of patients.  While privacy
violations are costly for the hospital as they result in reputation loss and
require expensive measures (such as privacy breach notifications), audit
inspections also cost money (e.g., the cost of the human auditor involved in
the investigation). Moreover, the number and type of privacy violations depend
on the actions of the \emph{rational} auditee---employees commit 
violations that benefit them.

\subsection{Our Model}

We model the audit process as a game between a defender (e.g, a hospital) and
an adversary (e.g., an employee). The defender audits a given set of
targets (e.g., health record accesses) and the adversary chooses a
target to attack.  
%The utility function of the defender captures the defender's
%incentives (e.g., loss from policy violations, cost of inspections). 
The defender's action space in the audit game includes two components. First,
the allocation of its inspection resources to targets; this component also
exists in a standard model of security games~\cite{tambe}.  Second, we
introduce a continuous punishment rate parameter that the defender employs to
deter the adversary from committing violations. However, punishments are not
free and the defender incurs a cost for choosing a high punishment level. 
For instance, 
 a negative work environment in a hospital with high
fines for violations can lead to a loss of productivity                 
(see~\cite{Becker68} for a similar account of the cost of punishment).  
The adversary's utility includes the benefit from committing            
violations and the loss from being punished if caught by the            
defender. Our model is parametric in the utility functions. Thus,       
depending on the application, we can instantiate the model to either    
allocate resources for detecting violations or preventing them.
This generality implies that our model can be used to       
study all the applications previously described in the security games     
literature~\cite{tambe}.                                                

%It may seem counter-intuitive that auditing prevents policy violations,
%as audits do  not stop violations in real time. However, setting appropriate
%punishments can deter the adversary from committing violations, leading to
%fewer violations. Additionally, refraining from auditing may result in
%violations caught externally (e.g., government audits, reports in media)
%leading to higher loss from violations; this incentivizes organizations to
%conduct audits. 

To analyze the audit game, we use the Stackelberg 
equilibrium solution concept~\cite{von1934marktform} in
which the defender commits to a strategy, and the adversary plays an optimal
response to that strategy. This concept captures situations in which 
the adversary learns the defender's audit strategy through surveillance or
the defender publishes its audit algorithm. 
In addition to yielding a better payoff for the defender than any Nash equilibrium,
the Stackelberg equilibrium makes the choice for the
adversary simple, which leads to a more predictable outcome of the game.
%---the adversary is often an employee and may not be able to solve the hard problem of
%computing a Nash equilibrium.  
Furthermore, this equilibrium concept respects the
computer security principle of avoiding ``security through obscurity''---
audit mechanisms like cryptographic algorithms should provide security despite being publicly
known. 

%\begin{table} \caption{Comparison: Security vs. Audit Games}
%\label{comparison} \begin{center} \begin{tabular}{c c} \toprule SG Term & AG
%Term\\ \midrule Defender/Attacker & Auditor/Auditee \\ Covered/Uncovered &
%Audited/Not audited \\ Resources & Budget \\ Schedules & (natural extension;
%not used in this work) \\ - & Punishment rate \\ \bottomrule \end{tabular}
%\end{center} \end{table}

\subsection{Our Results}
%In more details, the defender's action space in the audit game has a
%continuous component---the punishment rate. 
Our approach to computing the Stackelberg equilibrium is based on the multiple
LPs technique of Conitzer and Sandholm~\shortcite{ConitzerS06}. However, due to
the effect of the punishment rate on the adversary's utility, the optimization
problem in audit games has quadratic and non-convex constraints. The
non-convexity does not allow us to use any convex optimization methods, and in
general polynomial time solutions for a broad class of non-convex optimization
problems are not known~\cite{CambridgeJournals:227247}.

However, we demonstrate that we can efficiently obtain an additive
approximation to our problem. Specifically, we present an additive fully
polynomial time approximation scheme (FPTAS) to solve the audit game
optimization problem. Our algorithm provides a $K$-bit precise output in
time polynomial in $K$.  Also, if the solution is rational, our algorithm
provides an exact solution in polynomial time. In general, the
exact solution may be irrational and may not be representable in a finite
amount of time. 
%The additive approximation can be converted to a multiplicative approximation
%if the magnitude of the optimal answer can be lower bounded.  (but in general
%this is not true for optimization problems over real numbers). 

\subsection{Related Work}

Our audit game model is closely related to
security games~\cite{tambe}.  There are many papers (see,
e.g.,~\cite{KorzhykCP10,PitaTKCS11,PitaJOPTWPK08}) on security games, and as
our model adds the additional continuous punishment parameter, all
the variations presented in these papers can be studied in the context
of audit games (see Section~\ref{sec:disc}). However, the audit game solution
is technically more challenging as it involves non-convex constraints. 
%\ariel{As examples of papers on security games I would cite AAAI/IJCAI papers
%instead of AAMAS papers; look them up in Tambe's DBLP, and we can also cite
%the AAAI'10 complexity paper by Korzhyk et al.}

An extensive body of work on auditing focuses on analyzing logs
for detecting and explaining violations using techniques based on
logic~\cite{VaughanJMZ08,GargJD11} and machine
learning~\cite{Zheng06statisticaldebugging,Bodik2010}. In contrast,
very few papers study economic considerations in auditing strategic 
adversaries.  Our work is inspired in part by the model proposed in one such paper~\cite{BCDS12}, which also takes the point of view of commitment and
Stackelberg equilibria to study auditing. However, the emphasis in that work is 
on developing a detailed model and using it to predict observed audit practices 
in industry and the effect of public policy interventions on auditing practices. 
They do not present efficient algorithms for computing the optimal audit strategy. 
In contrast, we work with a more general and simpler model and present an efficient 
algorithm for computing an approximately optimal audit strategy. Furthermore,
since our model is related to the security game model, it opens up the possibility 
to leverage existing algorithms for that model and apply the results to 
the interesting applications explored with security games. 

%\ariel{Need to flesh out this paragraph; say a bit about the results in the
%BCDS12 paper.}

%The model of Blocki et al.~\shortcite{BCDS12} is in turn built on a slightly
%older model~\cite{BCDS11}, in which nothing was assumed about the adversary
%and hence regret minimization (see, e.g.,~\cite{blum2007learning}) was used as
%an auditing strategy.
%#######################################
% ADD BACK IN CAMERA READY VERSION
%#######################################
 
\citeauthor{ZhaoJ08}~\shortcite{ZhaoJ08} model a specific audit strategy---``break the glass''----where 
agents are permitted to violate an access control policy at their discretion
(e.g., in an emergency situation in a hospital), but these actions are 
audited. They manually analyze specific utility functions 
and obtain closed-form solutions for the audit strategy that results in a 
Stackelberg equilibrium. In contrast, our results apply to any utility function
and we present an efficient algorithm for computing the audit strategy.

%########################################

%Some techniques for non-convex constraints involves obtain convex relaxations
%of the same, and then bounding the 
%magnitude of the difference between the resultant solution and the true solution.

\section{The Audit Game Model} \label{notation}

The audit game features two players: the defender ($D$), and the        
adversary ($A$). The defender wants to audit $n$ targets $t_1, \ldots,  
t_n$, but has limited resources which allow for auditing only one of    
the $n$~targets. Thus, a pure action of the defender is to choose which 
target to audit.                                                        
A randomized strategy is a vector of probabilities 
$p_1, \ldots, p_n$ of each target being audited. The adversary attacks one 
target such that given the defender's strategy the adversary's choice 
of violation is the best response.  
%It is possible that adversary does not attack any target (no violation); to allow such a possibility
%we include a dummy target for which all associated costs are zero. 
%The dummy target forces the resource to be allocated to 
%audit some target, i.e., zero probability of not auditing.

Let the utility of the defender be $U^a_D (t_i)$ when audited target $t_i$ was 
found to be attacked, 
and $U^u_D (t_i)$ when unaudited target $t_i$ was found to be attacked.
The attacks (violation) on unaudited targets are discovered by 
an external source (e.g. government, investigative journalists,...).
Similarly, define the 
utility of the attacker as $U^a_A (t_i)$ when the attacked target $t_i$ is audited, 
and $U^u_A (t_i)$ when attacked target $t_i$ is not audited, excluding
any punishment imposed by the defender.
Attacks discovered externally are costly for the defender, thus, $U^a_D (t_i) > U^u_D (t_i)$.
Similarly, attacks not
discovered by internal audits are more beneficial to the attacker, 
and $U^u_A (t_i) > U^a_A (t_i)$. 

The model presented so far is identical to security games with singleton
and homogeneous schedules, and a single resource~\cite{KorzhykCP10}. The
additional component in audit games is punishment. The defender chooses
a punishment ``rate'' $x \in [0,1]$ such that if auditing detects an
attack, the attacker is fined an amount~$x$. However, punishment is not
free---the defender incurs a cost for punishing, 
e.g., for creating a fearful environment. For ease of exposition, we
model this cost as a linear function $ax$, where $a > 0$; however, our
results directly extend to any cost function polynomial in $x$. Assuming
$x \in [0,1]$ is also without loss of generality as utilities can be
scaled to be comparable to $x$. We do assume the punishment rate is
fixed and deterministic; this is only natural as it must correspond to a
consistent policy.

We can now define the full utility functions. Given probabilities 
$p_1, \ldots, p_n$ of each target being audited, the utility of the 
defender when target $t_*$ is attacked is 
$$
p_{*}U^a_D (t_*) + (1 - p_{*})U^u_D (t_*) - ax.
$$ 
The defender pays a fixed cost $ax$ regardless of the outcome. 
In the same scenario, the 
utility of the 
attacker when target $t_*$ is attacked is 
$$
p_{*}(U^a_A (t_*) - x) + (1 - p_{*})U^u_A (t_*).
$$
The attacker suffers the punishment $x$ only when attacking an audited target. 

\medskip
\noindent\textbf{Equilibrium.} 
The Stackelberg equilibrium solution involves a commitment by the defender to a strategy (with a possibly randomized allocation of the resource), followed by the best response of the adversary.
The mathematical problem involves solving multiple optimization 
problems, one each for the case when attacking $t_*$ is in fact the best response of 
the adversary. Thus, assuming $t_{*}$ is the best response of the adversary, the 
$*^{th}$ optimization problem $P_*$ in audit games is
$$
\begin{array}{llc}
\displaystyle\max_{p_{i},x} & p_{*}U^a_D (t_*) + (1 - p_{*})U^u_D (t_*) - ax \ ,&\\
\mbox{subject to}& \forall i \neq *. \ p_{i}(U^a_A (t_i) -x) + (1 - p_{i})U^u_A (t_i) & \\
%&  \forall i \neq *. \ p_{i}(U^a_A (t_i) -x) + (1 - p_{i})U^u_A (t_i)  &\\
&  \ \ \ \ \ \ \ \ \ \ \ \ \ \ \ 
\leq p_{*}(U^a_A (t_*) - x) + (1 - p_{*})U^u_A (t_*)  \ ,&\\
& \forall i. \ 0 \leq p_i \leq 1 \ , &\\
& \sum_i p_i = 1 \ , & \\
&  0 \leq x \leq 1 \ .
\end{array}
$$
The first constraint verifies that attacking $t_*$ is indeed a best response. The auditor then solves the $n$ problems $P_1, \ldots, P_n$ (which correspond to the cases where the best response is $t_1,\ldots,t_n$, respectively), and chooses the best solution among 
all these solutions to obtain the final strategy to be used for auditing. This is a generalization of the multiple LPs approach of Conitzer and Sandholm~\shortcite{ConitzerS06}.

\medskip
\noindent\textbf{Inputs.} The inputs to the above problem are specified in $K$ bit 
precision. Thus, the total length of all inputs is $O(nK)$. Also, the model can be made more 
flexible by including a dummy target for which all associated costs are zero (including 
punishment); such a target models the possibility that the adversary does not 
attack any target (no violation). Our result stays the same with such a dummy target, but,
an additional edge case needs to be handled---we discuss this case in a remark
at the end of Section~\ref{algorithm}.

\section{Computing an Audit Strategy}
Because the indices of the set of targets can be arbitrarily permuted,
without loss of generality we focus on one optimization problem
$P_n$ ($* = n$) from the multiple optimization problems presented
in Section~\ref{notation}. The problem has quadratic and non-convex
constraints. The non-convexity can be readily checked by writing the
constraints in matrix form, with a symmetric matrix for the quadratic
terms; this quadratic-term matrix is indefinite.

However, for a fixed $x$, the induced problem is a linear programming
problem. It is therefore tempting to attempt a binary search over values
of $x$. This na\"ive approach does not work, because the solution may
not be single-peaked in the values of $x$, and hence choosing the right 
starting point for the binary search is a difficult problem.
Another na\"ive approach is to discretize the interval $[0,1]$ into steps 
of $\epsilon'$,
solve the resultant LP for the $1/\epsilon'$ many discrete values
of $x$, and then choose the best solution. As an LP can be solved in
polynomial time, the running time of this approach is polynomial in
$1/\epsilon'$, but the approximation factor is at least $a \epsilon'$
(due to the $ax$ in the objective). Since $a$ can be as large as $2^K$,
getting an $\epsilon$-approximation requires $\epsilon'$ to be $2^{-K}
\epsilon$, which makes the running time exponential in $K$. Thus, this
scheme %is only a PTAS as opposed to the FPTAS scheme we present.
cannot yield an FPTAS. 

\subsection{High-Level Overview}

Fortunately, the problem $P_{n}$ has another 
property that allows for efficient methods. Let us rewrite $P_{n}$ 
in a more compact form.
Let 
$\Delta_{D,i} = U^a_D (t_i)  - U^u_D (t_i)$, $\Delta_{i} = 
U^u_A (t_i) - U^a_A (t_i)$ and $\delta_{i,j} = U^u_A (t_i) - U^u_A (t_j)$.
$\Delta_{D, i}$ and $\Delta_{i}$ are always positive, and
 $P_n$ reduces to:
$$
\begin{array}{llc}
\displaystyle\max_{p_{i},x} & p_{n}\Delta_{D,n} + U^u_D (t_n) - ax \ ,&\\
\mbox{subject to}& \forall i \neq n. \ p_i (-x - \Delta_i) + p_n (x + \Delta_n) + \delta_{i,n}  \leq 0 \ , & \\
%&  \forall i \neq n. \ p_i (-x - \Delta_i) + p_n (x + \Delta_n) + \delta_{i,n}  \leq 0 \ ,&\\
& \forall i. \ 0 \leq p_i \leq 1 \ ,&\\
& \sum_i p_i = 1 \ ,&\\
&  0 \leq x \leq 1 \ .&
\end{array}
$$

The main observation that allows for polynomial time 
approximation is that, at the optimal solution point, the quadratic constraints can be partitioned into 
a) those that are tight, and b) those in which the probability variables $p_i$ are 
zero (Lemma~\ref{equalityorzero}). Each quadratic constraint corresponding to $p_i$ can be characterized by the 
curve $p_n (x + \Delta_n) + \delta_{i,n} = 0$. The quadratic constraints are
thus parallel hyperbolic curves on the $(p_n,x)$ plane; see Figure~\ref{fig} for an illustration. 
The optimal values $p_n^o,x^o$
partition the constraints (equivalently, the curves): the 
constraints lying below the optimal value are tight, and in the constraints above the
optimal value the probability variable $p_i$ is zero (Lemma~\ref{zeroprob}).
 The partitioning allows a linear number of 
iterations in the search for the solution,
with each iteration assuming that the optimal 
solution lies between adjacent curves and then
 solving the sub-problem with equality quadratic constraints.

Next, we reduce the problem with equality quadratic constraints to a
problem with two variables, exploiting the nature of the constraints
themselves, along with the fact that the objective has only two
variables. The two-variable problem can be further reduced to a
single-variable objective using an equality constraint, and 
elementary calculus then reduces the problem to finding the roots
of a polynomial. Finally, we use known results to find 
%exact values of rational roots or 
approximate values of irrational roots.
 
\subsection{Algorithm and Main Result}\label{algorithm}
The main result of our paper is the following theorem:
\begin{theorem} \label{mainthm}
Problem $P_n$ can be approximated to an additive $\epsilon$ in time 
\newbound{$O(n^{7} K (\log  n) + 
n^5 \log(\frac{1}{\epsilon}))$}{$O(n^{5} K  + 
n^4 \log(\frac{1}{\epsilon}))$}
using the splitting circle method~\cite{schonhage1982fundamental} for approximating roots.
\end{theorem}

\begin{remark}
The technique of Lenstra et al.~\shortcite{LenstraLenstraLovasz1982} can be used to exactly compute rational roots. Employing it in 
conjunction with the splitting circle method yields a time bound \newbound{$O(
\max \{n^{16} K^3 (\log  n)^3, ~n^4 \log({1}/{\epsilon})\})$}{$O(\max \{n^{13} K^3, ~n^{5} K  + 
n^4 \log({1}/{\epsilon})\})$}. Also, 
this technique finds an exact optimal solution if the solution is rational.
\end{remark}

%\begin{remark}
%It is possible to choose any polynomial with a polynomial degree instead of the linear
%term $ax$ in the objective, and still obtain a polynomial time solution. The analysis follows
%exactly the same steps as for the linear case. In particular, for a constant degree, the
%time bound is same as in Theorem~\ref{mainthm} above.
%\end{remark}

Before presenting our algorithm we state two results about the
optimization problem $P_n$ that motivate the algorithm and are also used
in the correctness analysis. The proof of the first lemma is omitted due
to lack of space.

\begin{lemma} \label{equalityorzero}
Let $p_n^o, x^o$ be the optimal solution. Assume $x^o > 0$ and 
$p_n^o < 1$. Then, at
 $p_n^o, x^o$, for all $i \neq n$, either $p_i=0$ or
$p_n^o (x^o + \Delta_n) + \delta_{i,n} = p_i (x^o + \Delta_i)$, i.e., the $i^{th}$ quadratic constraint is tight.
\end{lemma}
\begin{lemma} \label{zeroprob}
Assume $x^o > 0$ and 
$p_n^o < 1$. Let $p_n^o (x^o + \Delta_n) + \delta  = 0$. If for some $i$,  
$\delta_{i,n} < \delta$ then $p_i = 0$. 
If for some $i$,  $\delta_{i,n} > \delta$ then 
$p_n^o (x^o + \Delta_n) + \delta_{i,n} = p_i (x^o + \Delta_i)$.
If for some $i$, $\delta_{i,n} = \delta$ then 
$p_i = 0$ and $p_n^o (x^o + \Delta_n) + \delta_{i,n} = p_i (x^o + \Delta_i)$.
\end{lemma}
\begin{proof}
The quadratic constraint for $p_i$ is 
$
p_n^o (x^o + \Delta_n) + \delta_{i,n} \leq p_i (x^o + \Delta_i) 
$.
By Lemma~\ref{equalityorzero}, either $p_i = 0$ or the constraint is tight.
If $p_n^o (x^o + \Delta_n) 
+ \delta_{i,n} < 0$, then, since $p_i \geq 0$ and $x^o + \Delta_i \geq 0$, 
the constraint cannot be tight. Hence, $p_i = 0$. If $p_n^o (x^o + \Delta_n) 
+ \delta_{i,n} > 0$, then, $p_i \neq 0$ or else with $p_i = 0$
the constraint is not satisfied. Hence the constraint is tight. The last case with 
$p_n^o (x^o + \Delta_n) + \delta_{i,n} = 0$ is trivial.
\end{proof}

%\begin{figure}[t]
%\centering
%\includegraphics[width=0.45\textwidth]{regions.pdf}
%\caption{Regions} \label{regions}
%\label{regions}
%\end{figure}

\begin{figure}[t]
\centering

\begin{tikzpicture}[scale=0.7]

\tikzstyle{blackdot}=[circle,draw=black,fill=black,thin,inner sep=0pt,minimum size=1.5mm]

%\filldraw [fill=lightgray] (1.4,0) ellipse (0.8 and 1.4);
%\filldraw [fill=lightgray] (2.4,2) ellipse (0.5 and 0.6);

%\node at (1.4,0.3) {\scriptsize{tight}};
%\node at (1.4,-0.3) {\scriptsize{constraints}};
%\node at (2.4,2) {\scriptsize{$p_i=0$}};

\draw [-latex] (0,-2.8) -- (0,4);
\draw [-latex] (-1.8,0) -- (5,0);

\draw [lightgray] (3,0) -- (3,3);
\draw [lightgray] (0,3) -- (3,3);
\draw [lightgray] (-1.5,-2.8) -- (-1.5,4);

\node at (-1.5,-0.2) {\scriptsize{$-\Delta_n$}};

\draw [dashed] (-0.8, 2.8) .. controls (0,1) and (1,0.5) .. (3.5,0.5);
\draw [dashed] (-0.8, 4) .. controls (0,1.5) and (1,1) .. (3.5,1);
\draw [dotted] (-0.3, 4) .. controls (0.3,2) and (1.3,1.5) .. (3.5,1.5);

\draw [dashed] (-0.8, -2.8) .. controls (0,-1) and (1,-0.5) .. (3.5,-0.5);
%\draw [gray] (-0.8, -4) .. controls (0,-1.5) and (1,-1) .. (3.5,-1);

\node at (5,-0.2) {\scriptsize{$x$}};
\node at (-0.2,4) {\scriptsize{$p_n$}};
\node at (-0.2,3) {\scriptsize{$1$}};
\node at (3,-0.2) {\scriptsize{$1$}};

\node at (4.1,0.5) {\scriptsize{$\delta=-1$}};
\node at (4.1,1) {\scriptsize{$\delta=-2$}};
\node at (4.1,1.5) {\scriptsize{$\delta=-3$}};
\node at (4,-0.5) {\scriptsize{$\delta=1$}};
%\node at (4,-1) {\scriptsize{$\delta=2$}};

\node at (1,1.8) [blackdot] {};
\node at (1.2,2.7) {\scriptsize{$(p^o_n,x^o)$}};
\draw [-latex] (1.2,2.5) -- (1.03,1.9);

\end{tikzpicture}

\caption{The quadratic constraints are partitioned into those below $(p^o_n,x^o)$ that are tight (dashed curves), and those above $(p^o_n,x^o)$ where $p_i=0$ (dotted curves).} 
\label{fig} 
\end{figure}
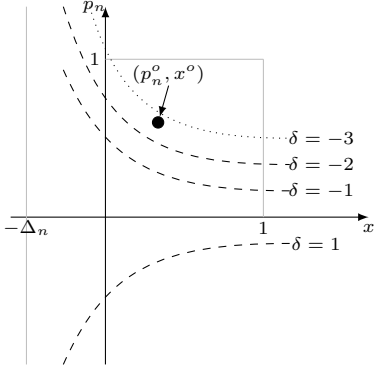

From Lemma~\ref{zeroprob}, if $p_n^o, x^o$ lies in the region between the adjacent
hyperbolas given by 
$p_n^o (x^o + \Delta_n) + \delta_{i,n} = 0$ and 
$p_n^o (x^o + \Delta_n) + \delta_{j,n} = 0$ (and $0 < x^o \leq 1$ and $0 \leq p_n^o < 1$), 
then $\delta_{i,n} \leq 0$ and $
p_i \geq 0$ and for the $k^{th}$ quadratic constraint with $\delta_{k,n} < \delta_{i,n}$, $p_k = 0$ and
for the $j^{th}$ quadratic constraint  with $\delta_{j,n} > \delta_{i,n}$, $p_j \neq 0$ and  
the constraint is tight. 

These insights lead to Algorithm~\ref{alg1}. After handling the case of $x=0$ and $p_n =1$ separately,
the algorithm sorts the $\delta$'s to get 
$\delta_{(1),n}, \ldots, \delta_{(n-1),n}$ in ascending order. Then, it iterates over 
the sorted $\delta$'s until a non-negative $\delta$ is reached, assuming the corresponding 
$p_i$'s to be zero and the other quadratic constraints to be equalities, and using the subroutine 
$\mathsf{EQ\_OPT}$ to solve the induced sub-problem. For ease of exposition we
assume $\delta$'s to be distinct, but the extension to repeated $\delta$'s is 
quite natural and does not require any new results. The sub-problem for the
$i^{th}$ iteration is given by the problem $Q_{n,i}$: 
$$
\begin{array}{llc}
\displaystyle\max_{x, p_{(1)}, \ldots, p_{(i)}, p_n} & p_{n}\Delta_{D,n} - ax \ ,&\\
\mbox{subject to}& p_n (x + \Delta_n) + \delta_{(i),n} \geq 0 \ ,& \\
%& p_n (x + \Delta_n) + \delta_{(i),n} \geq 0 \ ,&\\
& \mbox{if $i \geq 2$ then } p_n (x + \Delta_n) + \delta_{(i-1),n} < 0 \ ,&\\
&  \forall j \geq i. ~ p_n (x + \Delta_n) + \delta_{(j),n} = p_{(j)} (x + \Delta_j)  \ , &\\
& \forall j > i. ~ 0 < p_{(j)} \leq 1 \ ,&\\
& 0 \leq p_{(i)} \leq 1 \ ,&\\
& \sum_{k=i}^{n-1} p_{(k)} = 1 - p_n \ ,&\\
& 0 \leq p_n < 1 \ ,&\\
&  0 < x \leq 1  \ .&
\end{array}
$$
The best (maximum) solution from all the sub-problems (including $x=0$
and $p_n=1$) is chosen as the final answer. 

\begin{lemma}
\label{lem:om1}
Assuming $\mathsf{EQ\_OPT}$ produces an $\epsilon$-additive approximate objective value,
Algorithm~\ref{alg1} finds an $\epsilon$-additive approximate objective of optimization problem 
$P_n$.
\end{lemma}
The proof is omitted due to space constraints. 
\iffullversion
\begin{proof}
If $p_n^o (x^o + \Delta_n) + \delta_{i,n} \geq 0$ and 
$p_n^o (x^o + \Delta_n) + \delta_{j,n} < 0$, where $\delta_{j,n} < \delta_{i,n}$ and 
$\nexists k.~ \delta_{j,n} < \delta_{k,n} < \delta_{i,n}$, then the exact solution of 
the $i^{th}$ subproblem will be $p_n^o, x^o$. Now, since $0 < x \leq 1$ and 
$0 \leq p_n < 1$, there is one $i$ for which $p_n^o (x^o + \Delta_n) + \delta_{i,n} \geq 0$ and 
$p_n^o (x^o + \Delta_n) + \delta_{j,n} < 0$, and thus the solution of this sub-problem
will return the maximum value. The solution of other sub-problems will return a lower
value as the objective is same in all sub-problems. Hence, maximum
of the maximum in each iteration is the global maximum. The approximation case
is then an easy extension.
\end{proof}
\fi

\begin{algorithm}[h] \label{alg1}
\DontPrintSemicolon
$l \leftarrow prec(\epsilon, n, K)$, where $prec$ is defined after Lemma~\ref{thm1}\; 
%\nicolasc{Arunesh, please implement the change we discussed in the meeting}
Sort $\delta$'s in ascending order to get $\delta_{(1),n}, \ldots, \delta_{(n-1),n}$, with corresponding 
variables $p_{(1)}, \ldots, p_{(n-1)}$ and quadratic constraints $C_{(1)}, \ldots, C_{(n-1)}$\;
Solve the LP problem for the two cases when $x=0$ and $p_n=1$ respectively. Let 
the solution be $S^0, p^0_{(1)}, \ldots, p^0_{(n-1)}, p^0_n, x^0$ and 
$S^{-1}, p^{-1}_{(1)}, \ldots, p^{-1}_{(n-1)}, p^{-1}_n, x^{-1}$ respectively.\;
\For{$i\leftarrow 1$ \KwTo $n-1$}{
\eIf{$\delta_{(i),n} \leq 0 \lor (\delta_{(i),n} > 0 \land \delta_{(i-1),n} < 0)$}{
$p_{(j)} \leftarrow 0$ for $j < i$. \;
Set constraints $C_{(i)}, \ldots, C_{(n-1)}$ to be equalities. \;
$S^i, p^i_{(1)}, \ldots, p^i_{(n-1)}, p^i_n, x^i \leftarrow \mathsf{EQ\_OPT} (i,l)$\;
}{
$S^i \leftarrow -\infty$\;
}
}
$f \leftarrow \arg\max_i \{S^{-1}, S^0, S^1, \ldots, S^i, \ldots, S^{n-1} \}$ \;
$p^f_1, \ldots, p^f_{n-1} \leftarrow \mbox{Unsort } p^f_{(1)}, \ldots, p^f_{(n-1)}$ \;
\Return $p^f_1, \ldots, p^f_n, x^f$
\caption{$\mathsf{APX\_SOLVE}(\epsilon, P_n)$}
\end{algorithm}

$\mathsf{EQ\_OPT}$ solves a two-variable problem $R_{n,i}$ instead of 
$Q_{n,i}$. The problem $R_{n,i}$ is defined as follows:
$$
\begin{array}{llc}
\max_{x, p_{n}} ~~p_{n}\Delta_{D,n} - ax \ ,& &\\
\mbox{subject to}&& \\
 p_n (x + \Delta_n) + \delta_{(i),n} \geq 0 \ ,&&\\
 \mbox{if $i \geq 2$ then } p_n (x + \Delta_n) + \delta_{(i-1),n} < 0 \ ,&&\\
p_n \left ( 1 +  \sum_{j: i \leq j \leq n-1} \frac{x + \Delta_n}{x + \Delta_{(j)}} \right )
= 1 - \sum_{j: i \leq j \leq n-1} \frac{\delta_{{(j)},n}}{x + \Delta_{(j)}} \ ,&&\\
% \mbox{if $i \neq n-1$ then } \forall j \geq i. ~p_n (x + \Delta_n) + \delta_{{(j)},n} < (1 - p_n)(x + \Delta_{(j)})&&\\
 0 \leq p_n < 1 \ , &&\\
  0 < x \leq 1 \ .&&
\end{array}
$$
The following result justifies solving $R_{n,i}$ instead of $Q_{n,i}$.

\begin{lemma}
$Q_{n,i}$ and $R_{n,i}$ are equivalent for all $i$.
\end{lemma}

\begin{proof}
Since the objectives of both problems are identical, 
we prove that the feasible regions for the variables in the 
objective $(p_n, x)$ are identical. 
Assume $p_n, x, p_{(i)}, \ldots, p_{(n-1)}$ is feasible in 
$Q_{n,i}$. The first two constraints are the same in $Q_{n,i}$ and $R_{n,i}$. Divide 
each equality quadratic constraint corresponding to 
non-zero $p_{(j)}$ by $x +\Delta_{(j)}$. Add all such constraints to get:
$$
- \sum_{j: 1 \leq j \leq i} p_{(j)} + 
p_n \left ( \sum_{j: 1 \leq j \leq i} \frac{x + \Delta_n}{x + \Delta_{(j)}} \right ) + 
\sum_{j: 1 \leq j \leq i} \frac{\delta_{{(j)},n}}{x + \Delta_{(j)}} = 0
$$ Then, since $\sum_{k: 1 \leq k \leq i} p_{(k)} =
 1 - p_n$ we get 
$$
p_n \left ( 1 +  \sum_{j: i \leq j \leq n-1} \frac{x + \Delta_n}{x + \Delta_{(j)}} \right )
= 1 - \sum_{j: i \leq j \leq n-1} \frac{\delta_{{(j)},n}}{x + \Delta_{(j)}}.
$$
%Also, $\forall j \geq i. ~ p_n (x + \Delta_n) + \delta_{j,n} = 
%p_j (x + \Delta_j)$ and $p_j < 1 - p_n$ (when $i \neq n-1$ else $p_{n-1} = 1 - p_n$) implies $\forall j 
%\geq i. ~ p_n (x + \Delta_n) + \delta_{j,n} < (1 - p_n) (x^o + \Delta_j)$. 
The
last two constraints are the same in $Q_{n,i}$ and $R_{n,i}$.

Next, assume $p_n, x$ is feasible in $R_{n,i}$. Choose $$
p_{(j)} = 
p_n \left ( \frac{x + \Delta_n}{x + \Delta_{(j)}} \right ) + 
 \frac{\delta_{{(j)},n}}{x + \Delta_{(j)}} \ .
$$ Since $p_n (x + \Delta_n) + \delta_{(i),n} \geq 0$,
we have $p_{(i)} \geq 0$, and since $p_n (x + \Delta_n) + \delta_{(j),n} > 0$ for $j > i$
($\delta$'s are distinct)
we have $p_{(j)} > 0$. Also, 
$$ \sum_{j=i}^{n-1} p_{(j)} = p_n \left ( \sum_{j: i \leq j \leq n-1} \frac{x + \Delta_n}{x + \Delta_{(j)}} \right )
+ \sum_{j: i \leq j \leq n-1} \frac{\delta_{{(j)},n}}{x + \Delta_{(j)}},$$ 
which by the 
third constraint of $R_{n,i}$ is $1 - p_n$. This 
implies $p_{(j)} \leq 1$. Thus, $p_n, x, p_{(i)}, \ldots, p_{(n-1)}$ is feasible in 
$Q_{n,i}$.

%$p_n = \frac{x + \Delta_{(j)} - \delta_{(j),n}}
%{2x + \Delta_{(j)} + \Delta_n}= \frac{\Delta_{(j)} - 2\delta_{(j),n} - \Delta_n}{2x + \Delta_{(j)} + \Delta_n} + 1/2$.
%There are two cases, first  $\Delta_{(j)} - 2\delta_{(j),n} - \Delta_n > 0$.
\end{proof}

\begin{algorithm}[h] \label{alg2}
\DontPrintSemicolon
Define $F_i(x) = \frac{ 1 - \sum_{j: 1 \leq j \leq i-1} \frac{\delta_{j,n}}{x + \Delta_j}}
{1 +  \sum_{j: 1 \leq j \leq i-1} \frac{x + \Delta_n}{x + \Delta_j}}$\;
%Define $P(p_n,j,x) = p_n \displaystyle\frac{x + \Delta_n}{x + \Delta_{(j)}}
%+ \displaystyle\frac{\delta_{(j),n}}{x + \Delta_{(j)}}$\;
Define $feas(x) = \begin{cases}
     true & (x, F_i(x)) \mbox{ is feasible for } R_{n,i}\\
     false & \mbox{otherwise}
   \end{cases}$\;
Find polynomials $f,g$ such that $\frac{f(x)}{g(x)} = F_i(x) \Delta_{D,n} - ax$\;
$h(x) \leftarrow g(x) f'(x) - f(x) g'(x)$\;
$\{r_1, \ldots, r_s \} \leftarrow \mathsf{ROOTS}(h(x), l)$\;
$\{r_{s+1}, \ldots, r_{t} \} \leftarrow \mathsf{ROOTS}(F_i(x) + \frac{\delta_{(i),n}}
{x + \Delta_n}, l)$\; 
$\{r_{t+1}, \ldots, r_{u} \} \leftarrow \mathsf{ROOTS}(F_i(x), l)$\; 
%~\bigcup~ \mathsf{Roots}(F_i(x) + \frac{\delta_{{(i)},n}}{x + \Delta_n}) ~\bigcup~ \{ 1 \}
%~\bigcup~ \{\mathsf{Roots}(F_i(x) + \frac{\delta_{{(i-1)},n}}{x + \Delta_n}) \mbox{ if } i>1\} $\;
$r_{u+1} \leftarrow 1$\;
\For{$k \leftarrow 1$ \KwTo $u + 1$}{
\eIf{$feas(r_k)$}{
$O_{k} \leftarrow \frac{f(r_k)}{g(r_k)}$\;
}{
	\eIf{$feas(r_k - 2^{-l})$}{
	$O_{k} \leftarrow \frac{f(r_k - 2^{-l})}{g(r_k - 2^{-l})}$; $r_k \leftarrow r_k - 2^{-l}$\;
	}{
		\eIf{$feas(r_k + 2^{-l})$}{
		$O_{k} \leftarrow \frac{f(r_k + 2^{-l})}{g(r_k + 2^{-l})}$; $r_k \leftarrow r_k + 2^{-l}$\;
		}{
             $O_{k} \leftarrow -\infty$\;
             }
	}
}
}
$b \leftarrow \arg\max_k \{O_1, \ldots, O_k, \ldots, O_{u+1} \}$ \;

$ p_{(j)} \leftarrow 0$ for $j < i$\;
$p_{(j)} \leftarrow \displaystyle\frac{p_n(r_b + \Delta_n) + \delta_{(j),n}}{r_b + \Delta_{(j)}}$ for $j \in \{i, \ldots, n-1\}$\;
\Return $O_b, p_{(1)}, \ldots, p_{(n-1)}, p_n, r_b$
\caption{$\mathsf{EQ\_OPT} (i, l)$}
\end{algorithm}

The equality constraint in $R_{n,i}$, which forms a curve $K_i$,
allows substituting $p_n$ with a function $F_i(x)$ of the form ${f(x)}/{g(x)}$.
Then, the steps in $\mathsf{EQ\_OPT}$ involve taking the derivative of the objective 
${f(x)}/{g(x)}$ and finding those roots of the derivative that ensure that $x$ and $p_n$
satisfy all the constraints. The points with zero derivative are however local maxima only.
To find the global maxima, other values of $x$ of interest are where the curve 
$K_i$ intersects the \emph{closed} boundary of the region defined by the 
constraints. Only the closed boundaries are of interest, as maxima (rather
suprema) attained on open boundaries are limit points that are not
contained in the constraint region. However, such points are covered in
the other optimization problems, as shown below.

The limit point on the open boundary $p_n (x + \Delta_n) + \delta_{(i-1),n} < 0$
is given by the roots of $F_i(x) + \frac{\delta_{(i-1),n}}{x + \Delta_n}$. This point 
is the same as the point considered on the closed boundary 
$p_n (x + \Delta_n) + \delta_{(i-1),n} \geq 0$ in problem $R_{n,i-1}$ given by roots 
of $F_{i-1}(x) + \frac{\delta_{(i-1),n}}{x + \Delta_n}$, since 
$F_{i-1}(x) = F_i(x) $
when $p_n (x + \Delta_n) + \delta_{(i-1),n} = 0$. Also, the other cases when 
$x=0$ and $p_n=1$ are covered by the LP solved at the beginning of Algorithm~\ref{alg1}.

 The closed boundary in $R_{n,i}$ are obtained from the constraint $p_n (x + \Delta_n) 
+ \delta_{(i),n} \geq 0$, $0 \leq p_n$ and $x \leq 1$. The value $x$ of the intersection of $p_n (x + \Delta_n) 
+ \delta_{(i),n} = 0$ and $K_i$ is given by the roots of $F_i(x) + \frac{\delta_{(i),n}}
{x + \Delta_n} = 0$. The value  $x$ of the intersection of $p_n = 0$ and $K_i$ is 
given by roots of $F_i(x) = 0$. The value $x$ of the intersection of $x = 1$ and $K_i$ is 
simply $x = 1$. Additionally, as checked in $\mathsf{EQ\_OPT}$, 
all these intersection points must lie
with the constraint regions defined in $Q_{n,i}$.

The optimal $x$ is then the value among all the points of interest stated above 
that yields the maximum value for $\frac{f(x)}{g(x)}$. Algorithm~\ref{alg2}
describes $\mathsf{EQ\_OPT}$, which employs a root finding subroutine $\mathsf{ROOTS}$. Algorithm~\ref{alg2} also takes care of approximate results returned 
by the $\mathsf{ROOTS}$. As a result of the $2^{-l}$ approximation in the value of $x$, 
 the computed $x$ and $p_n$
can lie outside the constraint region when the actual $x$ and $p_n$ are very near the 
boundary of the region. Thus, we check for containment in the constraint region
for points $x \pm 2^{-l}$ and accept the point if the check passes. 
%Next, we analyze
%the guarantees provided by our algorithm.
\noindent \paragraph{Remark (dummy target):} As discussed in Section~\ref{notation},
we allow for a dummy target with all costs zero. Let this target be $t_0$. For $n$ not 
representing $0$, there is an
extra quadratic constraint given by $ p_0 (-x_0 - \Delta_0) + p_n (x + \Delta_n) + \delta_{0,n}  \leq 0$, 
but, as $x_0$ and $\Delta_0$ are $0$ the constraint is just $p_n (x + \Delta_n) + \delta_{0,n} \leq 0$.
When $n$ represents $0$, then the $i^{th}$ quadratic constraint is $p_i (-x - \Delta_i)  + \delta_{i,0}  \leq 0$, 
and the objective is independent of $p_n$ as $\Delta_{D,n} = 0$.  We first claim that $p_0 = 0$ at 
any optimal solution. The proof is provided in Lemma~\ref{dummy} in Appendix. Thus,
Lemma~\ref{equalityorzero} and~\ref{zeroprob} continue to hold for $i = 1 \mbox{ to } n-1$
with the additional
restriction that $p_n^o (x^o + \Delta_n) + \delta_{0,n} \leq 0$. 

Thus, when $n$ does not represent $0$, Algorithm~\ref{alg1}
runs with the the additional check $\delta_{(i),n} < \delta_{0,n}$ in the if
condition inside the loop. Algorithm~\ref{alg2} stays the same, except the additional
constraint that $p_0 = 0$. The other lemmas and
the final results stay the same. When $n$ represents $0$, then $x$ needs to be the smallest
possible, and the problem can be solved analytically.

\subsection{Analysis}
Before analyzing the algorithm's approximation guarantee we need a few results that we state below.
\begin{lemma} \label{prec}
The maximum bit precision of any coefficient of the polynomials given as input to $\mathsf{ROOTS}$
is \newbound{$3 n^2 K \log (n)$}{$2n(K + 1.5) + \log(n)$}.
\end{lemma}
\begin{proof}
The maximum bit precision will be obtained in $g(x) f'(x) - f(x) g'(x)$. Consider the 
worst case when $i = 1$. Then, $f(x)$ is of degree $n$ and $g(x)$ of degree $n-1$.
Therefore, the bit precision of $f(x)$ and $g(x)$ is upper bounded by 
\newbound{$nK \log({n \choose n/2})$}{$nK + \log({n \choose n/2})$}, where $nK$
comes from multiplying $n$ $K$-bit numbers and $\log({n \choose n/2})$ arises
from the maximum number of terms summed in forming any coefficient. Thus, using the 
fact that ${n \choose n/2} \leq (2e)^{n/2}$ the upper bound is approximately 
\newbound{$1.5 n^2 K$}{$n(K + 1.5)$}. 
We conclude that the bit precision of 
$g(x) f'(x) - f(x) g'(x)$ is upper bounded by 
\newbound{$(2 \cdot 1.5 n^2 K) \log (n) $}{$2n(K + 1.5) + \log(n)$}.
\end{proof}

We can now use Cauchy's result on bounds on root of polynomials to obtain a lower
bound for $x$. Cauchy's bound states that given
a polynomial $a_n x^n + \ldots + a_0$, any root $x$ satisfies  
$$
|x| > 1/\left(1 + \max
\{|a_n|/|a_0|, \ldots, |a_{1}|/|a_0|\}\right) \ .
$$
Using Lemma~\ref{prec} it can be concluded that any root returned by 
$\mathsf{ROOTS}$ satisfies 
\newbound{$x > 2^{- 6 n^2 K \log(n) - 1}$}{$x > 2^{- 4 n(K + 1.5) - 2\log(n) - 1}$}. 

Let \newbound{$B = 
2^{- 6 n^2 K \log(n) - 1}$}{$B = 2^{- 4 n(K + 1.5) - 2\log(n) - 1}$}. 
The following lemma (whose proof is omitted due to lack of space) bounds the additive approximation error. 

\begin{lemma} \label{approxcompute}
Assume $x$ is known with an additive accuracy of $\epsilon$, and
$\epsilon < B/2$. Then the error in the computed
$F(x)$ is at most $\epsilon \Psi$, where
$\Psi = \frac{Y + \sqrt{Y^2 + 4X}}{2}$ and 
$$
X =  \min  \Big \{ \displaystyle\mathop{\sum_{j: i \leq j \leq n-1,}}_{ \delta_{j,n} < 0}
\frac{ |\delta_{j,n}|}{(B + \Delta_j)^2}, \displaystyle\mathop{\sum_{j: i \leq j \leq n-1,}}_{ \delta_{j,n} > 0}
\frac{ 2\delta_{j,n}}{(B + \Delta_j)^2} \Big  \}
$$

$$
Y =  \min  \Big \{ \displaystyle\mathop{\sum_{j: i \leq j \leq n-1,}}_{ \Delta_{n} - \Delta_{j} < 0}
\frac{ |\Delta_{n} - \Delta_{j}|}{(B + \Delta_j)^2}, \displaystyle\mathop{\sum_{j: i \leq j \leq n-1,}}_{\Delta_{n} - \Delta_{j} > 0}
\frac{ 2(\Delta_{n} - \Delta_{j})}{(B + \Delta_j)^2} \Big  \}
$$

%$C$ or $D$ is of order $O(n2^{(12n^2 \log(n) + 1)K})$. Thus, 
Moreover, $\Psi$ is of 
order \newbound{$O(n2^{(12n^2 \log(n) + 1)K})$}
{$O(n2^{(8n(K + 1.5) +  4 \log(n) + K}) $}. 
%\ariel{I commented out a sentence here, please check that it makes sense.}
\end{lemma}

We are finally ready to establish the approximation guarantee of our algorithm.
\begin{lemma} \label{thm1}
Algorithm~\ref{alg1} solves problem $P_n$ with additive approximation term 
$\epsilon$ if 
$$
l > \max \{1 + \log(\frac{\Delta_{D,n} \Psi + a}{\epsilon}), 
\newbound{6 n^2 K \log (n) + 3}{4 n(K + 1.5) + 2\log(n) + 3} \}.
$$ 
Also, as 
$
\log(\frac{\Delta_{D,n} \Psi + a}{\epsilon}) = 
\newbound{O(n^2K\log(n) + \log(\frac{1}{\epsilon}))}{O(nK + \log(\frac{1}{\epsilon}))},
$ 
$l$ is of order $\newbound{O(n^2K\log(n) + \log(\frac{1}{\epsilon}))}{O(nK + \log(\frac{1}{\epsilon}))}$.
\end{lemma}

\begin{proof}
The computed value of $x$ can be at most $2 \cdot 2^{-l}$ far from the actual value. 
%\ariel{I don't understand this last sentence, by ``for'' do you mean ``times''?}
%\simina{does it make sense now?}
The additional factor of $2$ arises due to the boundary check in $\mathsf{EQ\_OPT}$.
Then using Lemma~\ref{approxcompute}, the maximum total additive approximation is
$2 \cdot 2^{-l} \Delta_{D,n} \Psi + 2 \cdot 2^{-l} a$. For this to be less than $\epsilon$,
$l > 1 + \log(\frac{\Delta_{D,n} \Psi + a}{\epsilon})$. The other term in the $\max$
above arises from the condition $\epsilon < B/2$ (this $\epsilon$ represents
$2 \cdot 2^{-l}$) in Lemma~\ref{approxcompute}.
\end{proof}

Observe that the upper bound on $\psi$ is only in terms of $n$ and $K$. Thus, we 
can express $l$ as a function of $\epsilon, n$ and $K$---$l = prec(\epsilon, n, K)$.

We still need to analyze the running time of the algorithm. First, we briefly discuss the 
known algorithms that we use and their corresponding running-time guarantees.
Linear programming can be done in polynomial time using Karmakar's algorithm~\cite{Karmarkar84}
with a time bound of $O(n^{3.5} L)$, where $L$ is the length of all inputs.

The splitting circle scheme to find roots of a polynomial combines
many varied techniques. The core of the algorithm yields linear polynomials $L_i = a_i x + b_i$
($a,b$ can be complex) such that the norm of the difference of the actual polynomial $P$ and the product
$\prod_{i} L_i$ is less than $2^{-s}$, i.e., $|P - \prod_{i} L_i| < 2^{-s}$. The 
norm considered is the sum of absolute values 
of the coefficient. The running time of the algorithm is $O(n^3 \log n + n^2s)$ in a 
pointer based Turing machine. By choosing $s = \theta(nl)$ and choosing the real part
of those complex roots that have imaginary value less than $2^{-l}$, 
it is possible to obtain approximations to the real 
roots of the polynomial with $l$ bit precision in time $O(n^3 \log n + n^3l)$. The 
above method may yield real values that lie near complex roots. However, such values
will be eliminated in taking the maximum of the objective over all real roots, 
if they do not lie near a real root.

% ################
% PUT BACK IN CAMERA READY
%
%\cut{
LLL~\cite{LenstraLenstraLovasz1982} is a method for finding a short basis of a given lattice. This is used to design polynomial time algorithms for factoring 
polynomials with rational coefficients into irreducible polynomials over rationals.
The complexity of this well-known algorithm is $O((n^{12} + n^9 (\log |f|)^3)$, 
when the polynomial is specified as in the field of integers and $|f|$ is the Euclidean 
norm of coefficients. For rational coefficients specified in $k$ bits, converting to 
integers yields $\log |f| \simeq \frac{1}{2}\log n + k$. LLL can be used before the 
splitting circle method to find all rational roots and then the irrational ones can be approximated.
%}
% #################
With these properties, we can state the following lemma. 

\begin{lemma}
The running time of Algorithm~\ref{alg1} with input approximation
parameter $\epsilon$ and inputs of $K$ bit precision is bounded by 
\newbound{$O(n^{7} K (\log  n) + 
n^5 \log(\frac{1}{\epsilon}))$}{$O(n^5 K + n^4 \log(\frac{1}{\epsilon}))$
}. 
% #### PUT BACK IN FINAL VERSION
% 
Using LLL yields the running time
\newbound{$O(
\max \{n^{16} K^3 (\log  n)^3, ~n^4 \log(\frac{1}{\epsilon})\})$}{
$O(\max \{n^{13} K^3, ~n^{5} K  + 
n^4 \log(\frac{1}{\epsilon})\})$}

% ##################
\end{lemma}

\begin{proof}
The length of all inputs is $O(nK)$, where $K$ is the bit precision of each constant.
The linear programs can be computed in time $O(n^{4.5} K)$.
The loop in Algorithm~\ref{alg1} runs less than $n$ times and calls $\mathsf{EQ\_OPT}$. In $\mathsf{EQ\_OPT}$, the computation happens
in calls to $\mathsf{ROOTS}$ and evaluation of the polynomial for each root found.
$\mathsf{ROOTS}$ is called three times with a polynomial of degree less than $2n$
and coefficient bit precision less than \newbound{$3n^2 K \log(n)$}{$2n(K + 1.5) + \log(n)$}. 
Thus, the total number of 
roots found is less than $6n$ and the precision of roots is $l$ bits. By Horner's method~\cite{Horner01011819},
polynomial evaluation can be done in the following simple manner: given a polynomial
$a_n x^n + \ldots + a_0$ to be evaluated at $x_0$ computing the following values 
yields the answer as $b_0$, $b_n = a_n$, $b_{n-1} = a_{n-1} + b_n x_0$, \ldots, 
$b_0 = a_0 + b_1 x_0$. 
From Lemma~\ref{thm1} we get $l \geq \newbound{3n^2 K \log(n)}{2n(K + 1.5) + \log(n)}$, thus, $b_i$ is approximately $(n+1 - i)l$ bits, and each computation involves multiplying two numbers
with less than $(n+1 - i)l$ bits each. We assume a pointer-based machine, thus multiplication is linear in number of bits. Hence the total time required for polynomial evaluation
is $O( n^2 l)$. The total time spent in all polynomial evaluation is
$O( n^3 l)$. The splitting circle method takes time
$O(n^3 \log n + n^3 l)$. Using Lemma~\ref{thm1} we get 
\newbound{$O(n^{6} K (\log  n) + 
n^4 \log(\frac{1}{\epsilon}))$}
{$O(n^4 K + n^3 \log(\frac{1}{\epsilon}))$} as the running time of 
$\mathsf{EQ\_OPT}$. Thus, the total time is $\newbound
{O(n^{7} K (\log  n) + 
n^5 \log(\frac{1}{\epsilon}))}{O(n^5 K + n^4 \log(\frac{1}{\epsilon}))}.$

When using LLL, the time in $\mathsf{ROOTS}$ in dominated by LLL. 
The time for LLL is given by \newbound{$O(n^{12} + n^9 (\log n + 3 n^2 K \log(n))^3)$}
{$O(n^{12} + n^9 (\log n + nK)^3)$},
which is 
\newbound{$O(n^{15} K^3 (\log  n)^3)$}{$O(n^{12} K^3)$}. 
Thus, the overall the time is bounded by 
\newbound{$O(\max \{n^{16} K^3 (\log  n)^3, ~n^4l)$}{$O(\max \{n^{13} K^3, ~n^4l)$}, 
which using Lemma~\ref{thm1} is
\newbound{$O(\max \{n^{16} K^3 (\log  n)^3, ~n^{6} K (\log  n) + 
n^4 \log(\frac{1}{\epsilon})\})$}{$O(\max \{n^{13} K^3, ~n^{5} K  + 
n^4 \log(\frac{1}{\epsilon})\})$}.
\end{proof}
% ################
% PUT BACK IN CAMERA READY
%

% ##################

\iffullversion
\begin{algorithm}[T] \label{alg3}
\DontPrintSemicolon
Use Sturm's sequence method to find the number of distinct real roots of h(x) in $(0,1]$, let it be $s$.\;
Use the LLL~\cite{} algorithm to factorize the polynomial in the field of rationals 
 to obtain $T$ rational roots $r_1 \ldots r_T$: 
$h(x) = (x - r_1) \cdot \ldots (x - r_T) \cdot m(x)$\;
$t \leftarrow $ number of distinct values in $(0,1]$ in $\{r_1, \ldots, r_T\}$\;
\If{$s \neq t$}{
Find the irrational roots of $m(x)$ with accuracy $2^{-l}$ using the splitting circle 
method~\cite{}.\;
}
Let $\{r_1, \ldots, r_s \}$ be all distinct real roots in $(0,1]$\;
Assume $r_i$ implicitly contains the type of root (rational/irrational), obtained using
function $type$\;
\Return $\{ r_1, \ldots, r_s \}$
\caption{$\mathsf{ROOTS}(h(x))$}
\end{algorithm}
\fi

\section{Discussion}
\label{sec:disc}

We have introduced a novel model of audit games that we believe to be compelling 
and realistic. Modulo the punishment parameter our setting reduces to the simplest 
model of security games. However, the security game framework is in general much 
more expressive. The model~\cite{Kiekintveld:2009:COR:1558013.1558108} 
includes a defender that controls multiple security 
resources, where each resource can be assigned to one of several \emph{schedules}, 
which are subsets of targets.
%\ariel{For the general model cite: Christopher Kiekintveld, Manish Jain, Jason Tsai, James Pita, Fernando Ordóñez, Milind Tambe: Computing optimal randomized resource allocations for massive security games. AAMAS (1) 2009: 689-696} 
For example, a security camera pointed in a specific direction monitors
all targets in its field of view. As audit games are also applicable in
the context of prevention, the notion of schedules is also relevant for
audit games.
%\ariel{Example? I don't want to sound ignorant here...}
Other extensions of security games are relevant to both prevention and 
detection scenarios, including an adversary that 
attacks multiple targets~\cite{KorzhykCP11}, and a defender with a budget~\cite{Bhattacharya:2011:AAS:2188743.2188746}.
 Each such extension raises difficult 
algorithmic questions. 

%##################################
%ADD BACK IN CAMERA READY VERSION
%##################################

%An approach of particular interest
%is to leverage existing computational results in security games,        
%and use those as black-box subroutines to obtain solutions for audit     
%games. Our paper is \emph{not} a good example of this approach, as the  
%problem of allocating a single resource with singleton schedules is a   
%simple LP, and the transition to the audit game solution is extremely   
%intricate; but we cannot rule out the existence of general black-box    
%reductions. 
%
%Apart from extensions based on security games, there are natural
%extensions based on variants of punishment, e.g., different punishments
%for different targets, or punishments that scale non-linearly with the
%number of attacked targets.  

%#################################                                                          

Ultimately, we view our work as a first step toward a computationally
feasible model of audit games. We envision a vigorous interaction
between AI researchers and security and privacy researchers, which
can quickly lead to deployed applications, especially given the
encouraging precedent set by the deployment of security
games algorithms~\cite{tambe}.

%Such extensions and open
%questions lays fertile ground for additional research in the area of audit games. Another
%open question is whether all solutions of the audit game problem are rational? If yes,
%our approach using LLL provides an exact solution in polynomial time, otherwise
%the best that can be done is an approximation. Also, it may be possible to use KKT 
%conditions~\cite{} to derive the results about the optimal solution (Lemma~\ref{equalityorzero},~\ref{zeroprob}), 
%however, our direct approach provides intuition about the results that leads to the
%algorithm.

%\ariel{Some of the references are incomplete, in particular Blum et al. and Sch\"onhage. As you know there are also some missing refs.}

\bibliographystyle{named}
\bibliography{Ijcai13}

\appendix

\section{Missing proofs}

\begin{proof}[Proof of Lemma~\ref{equalityorzero}] 
We prove the contrapositive. Assume there exists a $i$ such that 
$p_i \neq 0$ and the $i^{th}$ quadratic constraint 
is not tight. Thus, there exists an $\epsilon > 0$ such that
$$
p_i (-x^o - \Delta_i) + p_n^o (x^o + \Delta_n) + \delta_{i,n} + \epsilon = 0 \ .
$$
We show that it is possible to increase to $p_n^o$ by a small amount 
such that all constraints are satisfied, which leads to a higher 
objective value, proving that $p_n^o, x^o$ is not optimal. 
Remember that all $\Delta$'s are $\geq 0$, and $x > 0$. 

We do two cases: (1) 
assume  
$\forall l \neq n.~ p_n^o (x^o + \Delta_n) + \delta_{l,n} \neq 0$.
 Then, first, note that $p_n^o$ can be increased by $\epsilon_i$ or less and and $p_i$ can be 
decreased by $\epsilon_i'$ to still satisfy the constraint, as long as 
$$\epsilon_i' (x^o + \Delta_i) + \epsilon_i (x^o + \Delta_n) \leq \epsilon \ .$$ 
It is always possible to choose such $\epsilon_i > 0, \epsilon_i' > 0$.
Second, note that for those $j$'s for which $p_j = 0$ we get $p_n^o (x^o + \Delta_n) 
+ \delta_{j,n} \leq 0$, and by assumption $p_n^o (x^o + \Delta_n) + 
\delta_{j,n} \neq 0$, thus, $p_n^o (x^o + \Delta_n) + \delta_{j,n} < 0$. 
Let $\epsilon_j$ be such  that
$(p_*^o +  \epsilon_j)(x^o + \Delta_*) + \delta_{j,*} = 0$, i.e., $p_n^o$ 
can be increased by $\epsilon_j$ or less and the $j^{th}$ constraint will still be satisfied. Third, for 
those $k$'s for which $p_k \neq 0$, $p_n^o$ can be increased by $\epsilon_k$ or less, which 
must be accompanied with $\epsilon_k' = \frac{x^o + \Delta_*}{x^o + \Delta_i} 
\epsilon_k$ increase in $p_k$ in order to satisfy the $k^{th}$ quadratic constraint. 

Choose feasible $\epsilon_k'$'s (which fixes the choice of $\epsilon_k$ also) such that $\epsilon_i' - \sum_k \epsilon_k' > 0$. Then
choose an increase in $p_i$: $\epsilon_i'' < \epsilon_i'$ such that $$\epsilon_n = 
\epsilon_i'' - \sum_k \epsilon_k' > 0\mbox{ and  }\epsilon_n <
\min \{
\epsilon_i, \min_{p_j=0} \epsilon_j, \min_{p_k \neq 0} \epsilon_k \} 
$$ Increase
$p_n^o$ by $\epsilon_n$, $p_k$'s by $\epsilon_k'$ and decrease $p_i$ by $\epsilon_i''$ so that the constraint 
$\sum_i p_i = 1$ is still satisfied. Also, observe that choosing an increase 
in $p_n^o$ that is less than any $\epsilon_k$, any $\epsilon_j$, 
$\epsilon_i$ satisfies the quadratic constraints corresponding 
to $p_k$'s, $p_j$'s and $p_i$ respectively. 
Then, as $\epsilon_n > 0$ we have shown that $p_n^o$ cannot be optimal.

Next, for the case (2) if $p_n^o (x^o + \Delta_n) + \delta_{l,n} = 0$ for some $l$ then $p_l = 0$, 
, $\delta_{l,n} < 0$ and the objective becomes $$p_n \Delta_n -  
\frac{\delta_{l,n}}{p_n} - \Delta_n \ .$$ Thus, increasing $p_n$ increases the objective. Note
that choosing a lower than $x^o$ feasible value for $x$, results in an higher than 
$p_n^o$ value for $p_n$.
Also, the $k^{th}$ constraint can be written as 
$p_k(-x-\Delta_k) + \delta_{k,n} - \delta_{l,n} \leq 0$. We show that 
it is possible to choose a feasible $x$ lower than $x^o$. If for some $j$, $p_j = 0$, 
then $x$ can be decreased without violating the corresponding constraint. Let $p_t$'s
be the probabilities that are non-zero and let the number of such $p_t$'s be $T$.
By assumption there is an $i \neq l$ such that $p_i > 0$ and
$$
p_i (-x^o - \Delta_i) + \delta_{i,n} - \delta_{l,n} + \epsilon =  0 \ .
$$
For $i$, it is possible to decrease $p_i$ by $\epsilon_{i}$
such that $\epsilon_{i}(x^o + \Delta_i) \leq \epsilon/2$, hence 
the constraint remains satisfied and is still non-tight.

Increase each $p_t$ by $\epsilon_{i}/T$ so that the constraint $\sum_i p_i = 1$ is
satisfied. Increasing $p_t$ makes the $t^{th}$ constraint becomes 
non-tight for sure. Then, all constraints with probabilities greater than $0$ are 
non-tight. For each such constraint it is possible to decrease $x$ (note $x^o > 0$)  without violating
the constraint.Thus, we obtain a lower feasible $x$ than $x^o$, hence a higher $p_n$ 
than $p_n^o$. Thus, $p_n^o, x^o$ is not optimal.
\end{proof}

\begin{proof}[Proof of Lemma~\ref{lem:om1}]
If $p_n^o (x^o + \Delta_n) + \delta_{i,n} \geq 0$ and 
$p_n^o (x^o + \Delta_n) + \delta_{j,n} < 0$, where $\delta_{j,n} < \delta_{i,n}$ and 
$\nexists k.~ \delta_{j,n} < \delta_{k,n} < \delta_{i,n}$, then the exact solution of 
the $i^{th}$ subproblem will be $p_n^o, x^o$. Now, since $0 < x \leq 1$ and 
$0 \leq p_n < 1$, there is one $i$ for which $p_n^o (x^o + \Delta_n) + \delta_{i,n} \geq 0$ and 
$p_n^o (x^o + \Delta_n) + \delta_{j,n} < 0$, and thus the solution of this sub-problem
will return the maximum value. The solution of other sub-problems will return a lower
value as the objective is same in all sub-problems. Hence, maximum
of the maximum in each iteration is the global maximum. The approximation case
is then an easy extension.
\end{proof}

\begin{proof}[Proof of Lemma~\ref{approxcompute}]
$$
\frac{\Delta_n - \Delta_j} {x - \epsilon + \Delta_j} <
\frac{\Delta_n - \Delta_j}{x + \Delta_j} \mbox{ if $\Delta_n - \Delta_j < 0$}
$$
and using the fact that $\frac{1}{1 - \epsilon} < 1 + 2 \epsilon$ for $\epsilon < 1/2$,
$$
\frac{\Delta_n - \Delta_j} {x - \epsilon + \Delta_j} <
\frac{\Delta_n - \Delta_j}{x + \Delta_j} + 2\epsilon \frac{\Delta_n - \Delta_j}{(x + \Delta_j)^2} \mbox{ }
$$
if $\Delta_n - \Delta_j > 0$ and $\frac{\epsilon}{x + \Delta_j} < 1/2$. The latter
condition is true as $\epsilon < B/2$.
Thus,
$$
\begin{array} {lll}
\displaystyle\sum_{j: i \leq j \leq n-1} 
\frac{\Delta_n - \Delta_j} {x - \epsilon + \Delta_j} & < &
\displaystyle\sum_{j: i \leq j \leq n-1} 
\frac{\Delta_n - \Delta_j} {x + \Delta_j}  +  \\
& & \displaystyle\sum_{j: i \leq j \leq n-1, \Delta_n - \Delta_j > 0}
\frac{ 2\epsilon(\Delta_n - \Delta_j)}{(x + \Delta_j)^2}
\end{array}
$$
$$
\frac{\Delta_n - \Delta_j} {x + \epsilon + \Delta_j} <
\frac{\Delta_n - \Delta_j}{x + \Delta_j} \mbox{ if $\Delta_n - \Delta_j > 0$}
$$
and using the fact that $\frac{1}{1 + \epsilon} > 1 -  \epsilon$,
$$
\frac{\Delta_n - \Delta_j} {x + \epsilon + \Delta_j} <
\frac{\Delta_n - \Delta_j}{x + \Delta_j} - \epsilon \frac{\Delta_n - \Delta_j}{(x + \Delta_j)^2} \mbox{ }
$$
if $\Delta_n - \Delta_j < 0$.
Thus,
$$
\begin{array} {lll}
\displaystyle\sum_{j: i \leq j \leq n-1} 
\frac{\Delta_n - \Delta_j} {x + \epsilon + \Delta_j} & < &
\displaystyle\sum_{j: i \leq j \leq n-1} 
\frac{\Delta_n - \Delta_j} {x + \Delta_j}  +  \\
& & \displaystyle\sum_{j: i \leq j \leq n-1, \Delta_n - \Delta_j < 0}
\frac{ \epsilon|\Delta_n - \Delta_j|}{(x + \Delta_j)^2}
\end{array}
$$
Thus, using the fact that $x > B$ we have
$$
\begin{array} {lll}
\displaystyle\sum_{j: i \leq j \leq n-1} 
\frac{\Delta_n - \Delta_j} {x_{\epsilon} + \Delta_j} & < &
\displaystyle\sum_{j: i \leq j \leq n-1} 
\frac{\Delta_n - \Delta_j} {x + \Delta_j}  +  \\
& & \epsilon \min  \Big \{ \displaystyle\sum_{j: i \leq j \leq n-1, \Delta_n - \Delta_j < 0}
\frac{ |\Delta_n - \Delta_j|}{(B + \Delta_j)^2},\\
& & \displaystyle\sum_{j: i \leq j \leq n-1, \Delta_n - \Delta_j > 0}
\frac{ 2(\Delta_n - \Delta_j)}{(B + \Delta_j)^2} \Big  \} \\
& = & \displaystyle\sum_{j: i \leq j \leq n-1} 
\frac{\Delta_n - \Delta_j} {x + \Delta_j} + \epsilon Y
\end{array}
$$
%Similarly we get
%$$
%\begin{array} {lll}
%\displaystyle\sum_{j: i \leq j \leq n-1} 
%\frac{\Delta_n - \Delta_j} {x_{\epsilon} + \Delta_j} & > &
%\displaystyle\sum_{j: i \leq j \leq n-1} 
%\frac{\Delta_n - \Delta_j} {x + \Delta_j}  -  \\
%& & \epsilon \min  \Big \{ \displaystyle\sum_{j: i \leq j \leq n-1, \Delta_n - \Delta_j > 0}
%\frac{ (\Delta_n - \Delta_j)}{(B + \Delta_j)^2},\\
%& & \displaystyle\sum_{j: i \leq j \leq n-1, \Delta_n - \Delta_j < 0}
%\frac{ 2|\Delta_n - \Delta_j|}{(B + \Delta_j)^2} \Big  \} \\
%& = & \displaystyle\sum_{j: i \leq j \leq n-1} 
%\frac{\Delta_n - \Delta_j} {x + \Delta_j} - \epsilon D'
%\end{array}
%$$
Very similar to the above proof we also get
$$
\begin{array} {lll}
- \displaystyle\sum_{j: i \leq j \leq n-1} 
\frac{\delta_{j,n}} {x_{\epsilon} + \Delta_j} & > &
-\displaystyle\sum_{j: i \leq j \leq n-1} 
\frac{\delta_{j,n}} {x + \Delta_j}  -  \\
& & \epsilon \min  \Big \{ \displaystyle\sum_{j: i \leq j \leq n-1, \delta_{j,n} < 0}
\frac{ |\delta_{j,n}|}{(B + \Delta_j)^2},\\
& & \displaystyle\sum_{j: i \leq j \leq n-1, \delta_{j,n} > 0}
\frac{ 2\delta_{j,n}}{(B + \Delta_j)^2} \Big  \} \\
& = & -\displaystyle\sum_{j: i \leq j \leq n-1} 
\frac{\delta_{j,n}} {x + \Delta_j} - \epsilon X
\end{array}
$$
%and 
%$$
%\begin{array} {lll}
%- \displaystyle\sum_{j: i \leq j \leq n-1} 
%\frac{\delta_{j,n}} {x_{\epsilon} + \Delta_j} & < &
%-\displaystyle\sum_{j: i \leq j \leq n-1} 
%\frac{\delta_{j,n}} {x + \Delta_j}  +  \\
%& & \epsilon \min  \Big \{ \displaystyle\sum_{j: i \leq j \leq n-1, \delta_{j,n} > 0}
%\frac{ \delta_{j,n}}{(B + \Delta_j)^2},\\
%& & \displaystyle\sum_{j: i \leq j \leq n-1, \delta_{j,n} < 0}
%\frac{ 2 |\delta_{j,n}|}{(B + \Delta_j)^2} \Big  \} \\
%& = & -\displaystyle\sum_{j: i \leq j \leq n-1} 
%\frac{\delta_{j,n}} {x + \Delta_j} + \epsilon C'
%\end{array}
%$$
Then given
$$
p_n = \frac
{1 - \sum_{j: i \leq j \leq n-1} \frac{\delta_{{(j)},n}}{x + \Delta_{(j)}}}
{\left ( 1 +  \sum_{j: i \leq j \leq n-1} \frac{x + \Delta_n}{x + \Delta_{(j)}} \right )}
= \frac{A}{B}
$$
Using the inequalities above
$$
F(x_{\epsilon})= \frac
{1 - \sum_{j: i \leq j \leq n-1} \frac{\delta_{{(j)},n}}{x _{\epsilon}+ \Delta_{(j)}}}
{\left ( 1 +  \sum_{j: i \leq j \leq n-1} \frac{x_{\epsilon} + \Delta_n}{x_{\epsilon} + \Delta_{(j)}} \right )}
> \frac{A - \epsilon C}{B + \epsilon D}
$$
If $p_n \leq \psi$, then since $F(x_{\epsilon}) \geq 0$, we have $F(x) \geq p_n - \psi$.
If $p_n \geq \psi$, then since $B > 1$, $A > \psi$ we have
$$
\frac{A - \epsilon X}{B + \epsilon Y} > \frac{A}{B}(1 - \epsilon(\frac{X}{A} + \frac{Y}{B}))
> p_n(1  - \epsilon(\frac{X}{\psi} + Y))
$$
And, as $p_n < 1$ we have
$
F(x_{\epsilon}) > p_n - \epsilon(\frac{X}{\psi} + Y)
$.
Thus, $F(x_{\epsilon}) > p_n - \epsilon \min\{\psi, \frac{X}{\psi} + Y\}$.
The minimum is less that the positive root of $\psi^2 - Y \psi - X = 0$, that is
$\frac{Y + \sqrt{Y^2 + 4X}}{2}$.
\end{proof}

\section{Dummy Target}
\begin{lemma} \label{dummy}
If a dummy target $t_0$ is present for the optimization problem 
described in Section~\ref{notation}, then $p_0 = 0$ at optimal point $p_n^o, x^o$, where
$x^o > 0$.
\end{lemma}
\begin{proof}
We prove a contradiction. Let $p_0 > 0$ at optimal point. 
If the problem under consideration is when $n$ represents $0$, the objective
is not dependent on $p_n$ and thus, then we want to choose
$x$ as small as possible to get the maximum objective value. The quadratic inequalities 
are $p_i (-x^o - \Delta_i)  + \delta_{i,0}  \leq 0$. Subtracting $\epsilon$ from $p_0$ and
adding $\epsilon/n$ to all the other $p_i$, satisfies the $\sum_i p_i = 1$ constraint.
But, adding  $\epsilon/n$ to all the other $p_i$, allows reducing $x^o$ by a small amount
and yet, satisfying each quadratic constraint. But, then we obtain a higher objective 
value, hence $x^o$ is not optimal.

Next, if the problem under consideration is when $n$ does not represents $0$, then an 
extra constraint is $ p_n^o (x^o + \Delta_n) + \delta_{0,n}  \leq 0$.  Subtracting $\epsilon$ from $p_0$ and
adding $\epsilon/(n-1)$ to all the other $p_i$ (except $p_n^o$), satisfies the $\sum_i p_i = 1$ constraint. Also, 
each constraint $ p_i (-x^o - \Delta_i) + p_n^o (x^o + \Delta_n) + \delta_{i,n}  \leq 0$ becomes 
non-tight (may have been already non-tight) as a result of increasing $p_i$. Thus, 
now $x^o$ can be decreased (note $x^o > 0$). Hence, the objective
increases, thus $p_n^o, x^o$ is not optimal.
\end{proof}

\end{document}